\documentclass[a4paper]{article}

\usepackage[scale=0.8]{geometry}
\usepackage[utf8]{inputenc}
\usepackage[T2A]{fontenc}
\usepackage[russian,english]{babel}
\usepackage{indentfirst}

\usepackage{amsmath,amsthm,amssymb}

\usepackage{authblk}

\usepackage[table]{xcolor}
\usepackage[colorlinks=true,linkcolor=black,citecolor=black,urlcolor=blue]{hyperref}

\newtheorem{theorem}{Theorem}
\numberwithin{equation}{section}
\theoremstyle{plain}
\newtheorem{claim}[theorem]{Claim}

\theoremstyle{definition}
\theoremstyle{remark}
\newtheorem*{remark*}{Remark}

\newcommand{\tr}{\mathop{\operatorname{tr}}}
\newcommand{\id}{\mathop{\operatorname{id}}}

\newcommand{\colk}{\cellcolor{gray!20}}
\newcommand{\colr}{\cellcolor{rgb:red,10;white,30;blue,2}}
\newcommand{\colb}{\cellcolor{blue!25}}
\newcommand{\colc}{\cellcolor{cyan!20}}
\newcommand{\colm}{\cellcolor{green!30}}
\newcommand{\coly}{\cellcolor{yellow!35}}
\definecolor{orang}{RGB}{255,203,122}
\newcommand{\colg}{\cellcolor{orang}}

\begin{document}

\title{Strassen's $2 \times 2$ matrix multiplication algorithm:\\A conceptual perspective}

\author[1]{Christian Ikenmeyer}
\author[2]{Vladimir Lysikov}
\affil[1]{Max Planck Institute for Software Systems, Saarland Informatics Campus}
\affil[2]{Department of Computer Science, Saarland University}
\affil[ ]{\small\href{mailto:cikenmey@mpi-sws.org}{\texttt{cikenmey@mpi-sws.org}}\quad \href{mailto:vlysikov@cs.uni-saarland.de}{\texttt{vlysikov@cs.uni-saarland.de}}}

\date{\today}

\maketitle

\begin{abstract}
The main purpose of this paper is pedagogical.
  
Despite its importance, all proofs of the correctness of Strassen's famous 1969
algorithm to multiply two $2 \times 2$ matrices with only seven multiplications
involve some basis-dependent calculations such as explicitly multiplying
specific $2 \times 2$ matrices, expanding expressions to cancel terms with
opposing signs, or expanding tensors over the standard basis.
This makes the proof nontrivial to memorize and many presentations of
the proof avoid showing all the details and leave a significant amount of
verifications to the reader.

In this note we give a short, self-contained, basis-independent
proof of the existence of Strassen's algorithm that avoids these types of
calculations.
We achieve this by focusing on symmetries and algebraic properties.

Our proof can be seen as a coordinate-free version of the construction
of Clausen from 1988, combined with recent work on the geometry of Strassen's
algorithm by Chiantini, Ikenmeyer, Landsberg, and Ottaviani from 2016.
\end{abstract}

\section{Introduction}

The discovery of Strassen's matrix multiplication algorithm~\cite{Strassen:69}
was a breakthrough result in computational linear algebra.
The study of fast (subcubic) matrix multiplication algorithms initiated
by this discovery has become an important area of research
(see~\cite{Blaeser:13} for a survey and~\cite{LeGall:14} for the currently
best upper bound on the complexity of matrix multiplication).
Fast matrix multiplication has countless applications as a subroutine
in algorithms for a wide variety of problems,
see e.g.~\cite[\S 16]{BCS:97} for numerous applications in computational linear algebra.
In practice, algorithms more sophisticated than Strassen's are almost never
implemented, but Strassen's algorithm is used for multiplication of large
matrices
(see~\cite{Dumas-Pan:16,Pan:17,HRMvdG:17} on practical fast
matrix multiplication).

The core of Strassen's result is an algorithm for multiplying
$2 \times 2$ matrices with only $7$ multiplications instead of $8$.
It is a \emph{bilinear} algorithm, which means that it arises from a
decomposition of the form
\begin{equation}
\label{eq:bilinear}
  XY = \sum_{k = 1}^7 u_k(X) v_k(Y) W_k, \tag{$\star$}
\end{equation}
where $u_k$ and $v_k$ are cleverly chosen linear forms on the space of $2 \times 2$
matrices and $W_k$ are seven explicit $2 \times 2$ matrices.
Because of this structure it can be applied to block matrices, and its
recursive application results in an algorithm for the multiplication of two $n \times n$ matrices
using $O(n^{\log_2 7})$ arithmetic operations
(see~\cite[\S 15.2]{BCS:97} or~\cite{Blaeser:13} for details).

Because of the great importance of Strassen's algorithm,
our goal is to understand it on a deep level.
In Strassen's original paper, the linear forms $u_k$, $v_k$, and the matrices $W_k$ are given,
but the verification of the correctness of the algorithm is left to the reader.
Unfortunately, such a description does not yield many further immediate insights.

Shortly after Strassen's paper, Gastinel~\cite{Gastinel:71} published a proof
of the existence of decomposition~\eqref{eq:bilinear} using simple algebraic
transformations that is much easier to follow and verify.
Many other papers provide alternative descriptions of Strassen's algorithm or
proofs of its existence.
Brent~\cite{Brent:70} and Paterson~\cite{Paterson:74} present the algorithm
in a graphical form using $4 \times 4$ diagrams indicating which elements
of the two matrices are used.
A more formal version of these diagrams are matrices of linear forms, which are used,
for example, by Fiduccia~\cite{Fiduccia:72} (essentially the same proof
appears in~\cite{Yuval:78}), Brockett and Dobkin~\cite{Brockett-Dobkin:73}
and Lafon~\cite{Lafon:75}.
Makarov~\cite{Makarov:75} gives a proof
that uses ideas of Karatsuba's algorithm for the efficient multiplication of polynomials.
B{\"u}chi and Clausen~\cite{Buechi-Clausen:85} connect the existence of
Strassen's algorithm to the existence of special bases of the space of
$2 \times 2$ matrices in which the multiplication table has a specific
structure (their results are more general and apply not only to matrix
multiplication).
Alexeyev~\cite{Alekseyev:96} describes several algorithms for matrix
multiplication as embeddings of the matrix algebra into a $7$-dimensional
nonassociative algebra with a special properties.

Verification of these proofs usually requires simple, but lengthy computations:
expansion of explicit decompositions in some basis, multiplication of several
matrices or following chains of algebraic transformations in which careful attention to details is required.
To obtain a more conceptual proof of the existence of Strassen's algorithm,
we do not focus on the explicit algorithm, but on the algebraic
properties of the $2 \times 2$ matrices, their transformations and 
symmetries of Strassen's algorithm.
It is well-known that the decomposition~\eqref{eq:bilinear} is not unique.
Given one decomposition, we can obtain another one by applying the identity
\[
XY = A^{-1} \left[ (A X B^{-1}) (B Y C^{-1}) \right] C
\]
and using the original decomposition for the product in the square brackets.
Alternatively, we can talk about $2 \times 2$ matrices as linear maps between
$2$-dimensional vector spaces.
Any choice of bases in these vector spaces gives a new bilinear algorithm.
De Groote~\cite{deGroote:78b} proved that the algorithm with seven
multiplications is unique up to these transformations (this result is also
announced without a proof in~\cite{Pan:72}, see also~\cite{Pan:14}).
Thus, Strassen's algorithm is unique in this sense and there should be
a coordinate-free description of this algorithm which does not use explicit
matrices.
One such description is given in~\cite{CILO:17} and the proof of its
correctness uses the fact that matrix multiplication is the unique
(up to scale) bilinear map invariant under the transformations described above.
This is a nontrivial fact which requires representation theory to prove.
Moreover, the verification of the correctness in~\cite{CILO:17} is left to the reader.

Symmetries of Strassen's algorithm are also useful for its understanding.
Clausen~\cite{Clausen:88} gives a description of Strassen's algorithm in
terms of special bases, as in~\cite{Buechi-Clausen:85}, and notices that
the elements of these bases form orbits under the action of the symmetric group
$S_3$ on the space of $2 \times 2$ matrices defined via conjugation with specific
matrices, i.\,e., Strassen's algorithm is invariant under this action.
Clausen's construction is also describled in~\cite[Ch.1]{BCS:97}.
Grochow and Moore~\cite{Grochow-Moore:16,Grochow-Moore:17} generalize Clausen's
construction to $n \times n$ matrices using other finite group orbits.
Another symmetry is only apparent in the trilinear representation of the
algorithm: the decompositions~\eqref{eq:bilinear} are in one-to-one
correspondence with decompositions of the trilinear form $\tr(XYZ)$ of the form
\begin{equation*}
\label{eq:trilinear}
  \tr(XYZ) = \sum_{k = 1}^7 u_k(X) v_k(Y) w_k(Z)
\end{equation*}
where $u_k$, $v_k$ and $w_k$ are linear forms.
The decomposition corresponding to Strassen's algorithm is then invariant
under the cyclic permutation of matrices $X, Y, Z$.
This symmetry is exploited in the proof of Chatelin~\cite{Chatelin:86},
which uses properties
of polynomials invariant under this symmetry.
He also notices the importance of a matrix which is related to the $S_3$
symmetry discussed above.
The symmetries of Strassen's algorithm are explored in detail
in~\cite{Burichenko:14,CILO:17}. Several earlier publications note their
importance~\cite{Gates-Kreinovich:01,Paterson:09}.
The paper~\cite{BILR:17} explores symmetries of algorithms for $3 \times 3$
matrix multiplication.

In this paper we provide a proof of Strassen's result which is
\begin{itemize}
\item
  \emph{coordinate-free} --- we do not use explicit matrices, which allows us
  to focus on the algebraic properties required to prove the correctness of
  the algorithm. We avoid all tedious explicit calculations, in particular any
  expansions of expressions and any verification of explicit sign cancellations.
  Our proof can be seen as a coordinate-free version of Clausen's construction.
\item
  \emph{elementary} --- our proof uses only simple facts from basic linear
  algebra and does not require knowledge of representation theory.
  This is also why we do not use tensor language.
  Proofs from~\cite{CILO:17} and~\cite{Grochow-Moore:17} are based on more
  complicated mathematics and may offer other insights.
\end{itemize}

Formally, the result that we prove is the following.
\begin{theorem}[Strassen \cite{Strassen:69}]\label{thm:main}
Fix any field $\mathbb F$.
There exist
fourteen linear forms $u_1,\ldots,u_7, v_1,\ldots,v_7 \colon \mathbb F^{2 \times 2} \to \mathbb F$
and seven matrices $W_1,\ldots, W_7 \in \mathbb F^{2 \times 2}$ such that
for all pairs of $2 \times 2$ matrices $X$ and $Y$
the product satisfies
\[
  XY = \sum_{k=1}^7 u_k(X) v_k(Y) W_k. \tag{$\star$}
\]
\end{theorem}

\paragraph{\bf Acknowledgements.}
The authors thank Alin Bostan, Joshua Grochow and anonymous referees
for comments and pointers to the literature.

\section{Preliminaries from linear algebra}
The \emph{trace} $\tr(A)$ of a square matrix $A$ is the sum of its diagonal entries.
If $\tr(A)$ is zero, then the matrix $A$ is called \emph{traceless}.
Taking the trace of a product of (rectangular) matrices is invariant under cyclic shifts: $\tr(A_1 A_2 \cdots A_n) = \tr(A_2 \cdots A_n A_1)$.
As a consequence, the trace of a matrix is invariant under conjugations: $\tr(B^{-1}AB) = \tr(ABB^{-1}) = \tr(A)$.
Another implication is that if $u$ is a column vector and $v^T$ is a row vector, then $v^T u = \tr(v^T u) = \tr(u v^T)$.

The \emph{characteristic polynomial} of a $2 \times 2$ matrix $A$ is $\lambda^2 - \tr(A)\lambda + \det(A)$.
The Cayley---Hamilton theorem says that substituting $A$ for $\lambda$ yields the zero matrix.

\section{Rotational symmetry}\label{sec:rotationalsymmetry}
In this section we collect some standard facts about rotation matrices.
We think of the $2 \times 2$ matrix $D$ as a rotation of the plane by $120^\circ$,
but to make our approach work over every field we use a more algebraic definition for $D$.

Let $D$ have determinant $1$ and trace $-1$, that is, $D$ has characteristic polynomial \mbox{$\lambda^2 + \lambda + 1$}.
We assume that $D$ is not a multiple of the identity $\id$ (this is implicitly satisfied if the characteristic is not 3).
For example, we could choose $D = \begin{bmatrix}
                      0 & -1 \\
                      1 & -1
                      \end{bmatrix}$,
the matrix that cyclically permutes the three vectors $\begin{pmatrix}1\\0\end{pmatrix}$, $\begin{pmatrix}0\\1\end{pmatrix}$, $\begin{pmatrix}-1\\-1\end{pmatrix}$.

\begin{claim}\label{claim:sum}
  The matrix $D$ has the following properties: $D^3 = \id$ and $\id + D + D^{-1} = 0$ and $\tr(D^{-1})=-1$.
\end{claim}
\begin{proof}
The characteristic polynomial of $D$ is $\lambda^2 + \lambda + 1$. By the Cayley---Hamilton theorem $D^2 + D + \id = 0$.
Multiplying by $D$ we obtain $D+D^2+D^3=0=\id+D+D^2$ and hence $D^3 = \id$.
Therefore $D^{-1}=D^2$ and thus $\id + D + D^{-1} = 0$.
This implies $\tr(D^{-1}) = - \tr(\id) - \tr(D) = -1$.
\end{proof}

For every column vector $u$ define $u^{\perp}$ as the row vector satisfying conditions $u^{\perp} u = 0$ and $u^{\perp} D u = 1$.
If $u$ is not an eigenvector of $D$, then $u$ and $Du$ are linearly independent, so $u^{\perp}$ is uniquely defined.
If, on the other hand, $u$ is an eigenvector of $D$, the two conditions are inconsistent and $u^{\perp}$ is undefined.

We fix a vector $u$ that is not an eigenvector of $D$ and define $u^{\perp}$ as above.
In our example we could choose $u=\begin{pmatrix}1\\0\end{pmatrix}$, which is not an eigenvector of
$\begin{bmatrix}
  0 & -1 \\
  1 & -1
\end{bmatrix}$.

A first simple observation relates $u^\perp$ and $(Du)^{\perp}$:

\begin{claim}\label{claim:perp}
  $u^\perp D^{-1} = (D u)^\perp$.
\end{claim}
\begin{proof}
We need to verify the two defining properties for $(D u)^\perp$. We have $(u^\perp D^{-1})(D u) = u^\perp u = 0$ and $(u^\perp D^{-1}) D (D u) = u^\perp D u = 1$ as required.
\end{proof}

The following observation complements the fact that $u^\perp D u =1$.

\begin{claim}\label{claim:xdm}
$u^\perp D^{-1} u=-1$.
\end{claim}
\begin{proof}
  Using Claim~\ref{claim:sum} we have $\id+D+D^{-1}=0$ and thus \[u^\perp u+u^\perp Du+u^\perp D^{-1}u=0.\]
  Since $u^\perp u = 0$ and $u^\perp D u = 1$, the claim follows.
\end{proof}

\section{Seven multiplications suffice}
In this section we apply our structural insights from Section~\ref{sec:rotationalsymmetry} to prove Theorem~\ref{thm:main}.
We set $M := u u^\perp$. Clearly $\tr(M) = u^\perp u = 0$
and we obtain the following identities that can be used to simplify products of $M$, $D$, and $D^{-1}$:
\begin{claim}\label{claim:xdx}
$M^2 = 0$ and $MDM = M$ and $M D^{-1} M = -M$.
\end{claim}
\begin{proof}
\begin{eqnarray*}
M^2 &=& (u u^\perp) (u u^\perp) = u (u^\perp u) u^\perp = 0.\\
MDM &=& (u u^\perp) D (u u^\perp) = u (u^\perp D u) u^\perp = u u^\perp = M.\\
MD^{-1}M &=& (u u^\perp) D^{-1} (u u^\perp) = u (u^\perp D^{-1} u) u^\perp = -u u^\perp = -M,
\end{eqnarray*}
where in the last line we used Claim~\ref{claim:xdm}.
\end{proof}
By Claim~\ref{claim:sum}, conjugation with $D$ is a map of order $3$ on the vector space of all $2 \times 2$ matrices, i.e., for any matrix $A$ there is a triple of conjugates $A \mapsto D^{-1}AD \mapsto DAD^{-1} \mapsto A$.
Moreover, if $A$ is traceless, then so are its conjugates.

\begin{claim}\label{claim:conj-indep}
   The matrices $M$, $D^{-1}MD$, and $DMD^{-1}$ form a basis of the vector space of traceless matrices.
\end{claim}
\begin{proof}
Since $M$ is traceless, its conjugates are also traceless.
Hence it is enough to prove that $M$, $D^{-1}MD$ and $DMD^{-1}$ are linearly independent.

  Since $u$ is not an eigenvector of $D$, the vectors $u$ and $Du$ are linearly independent and thus form a basis of the space of column vectors.
  The row vectors $u^{\perp}$ and $u^{\perp} D^{-1} = (Du)^{\perp}$ (Claim~\ref{claim:perp}) are orthogonal to $u$ and $Du$, respectively.
  Therefore they form a basis of the space of row vectors.
  Thus, the four matrices
  \[
    u \cdot u^\perp = M,\quad u \cdot u^\perp D^{-1} = MD^{-1},\quad D u \cdot u^\perp = DM, \quad D u \cdot u^\perp D^{-1} = DMD^{-1}
  \]
  form a basis of the space of $2 \times 2$ matrices.
  The matrices $M$ and $DMD^{-1}$ are contained in this basis.
  Adding up all four matrices, we get $(\id + D) M (\id + D^{-1})$, which can be simplified to $(-D^{-1}) M (-D) = D^{-1}MD$ using Claim~\ref{claim:sum}.
  Therefore the matrices $M$, $DMD^{-1}$, $D^{-1}MD$ are linearly independent.
\end{proof}

Since $D$ and $D^{-1}$ have trace $-1 \neq 0$ (Claim~\ref{claim:sum}),
adding $D$ or $D^{-1}$ to the basis in Claim~\ref{claim:conj-indep} yields two bases for the full space of $2 \times 2$ matrices: $\{ D, M, D^{-1}MD, DMD^{-1} \}$ and $\{ D^{-1}, M, D^{-1}MD, DMD^{-1} \}$.

Using the properties $D^2 = D^{-1}$, $D^{-2} = D$ and $M^2 = 0$ from Claim~\ref{claim:sum} and Claim~\ref{claim:xdx}, we can write down the multiplication table with respect to these two bases.
We further simplify it using the identities $MDM = M$ and $MD^{-1}M = -M$ from Claim~\ref{claim:xdx}.

\begin{center}
\begin{tabular}{|c|c|c|c|c|}
  \hline
             & $D^{-1}$  & $M$         & $D^{-1}MD$  & $DMD^{-1}$             \\
  \hline
  $D$        & \colk $\id$     & \colm $DM$        & \colc $MD$        & \colb \begin{minipage}{1.65cm}\medskip$D^{-1}MD^{-1}$\medskip\end{minipage}        \\
  \hline
  $M$        & \colr $MD^{-1}$ & $0$         & \colc \begin{minipage}{1.65cm}{$MD^{-1}MD$ \\$= -MD$}\end{minipage} & \colr \begin{minipage}{1.6cm}{\medskip$MDMD^{-1}$ \\ $= MD^{-1}$}\medskip\end{minipage}  \\
  \hline
  $D^{-1}MD$ & \colg $D^{-1}M$ & \colg \begin{minipage}{1.6cm}{$D^{-1}MDM$\\ $= D^{-1}M$}\end{minipage} & $0$         & \colb \begin{minipage}{2.7cm}{\medskip$D^{-1}MD^{-1}MD^{-1}$ \\ $ = -D^{-1}MD^{-1}$}\medskip\end{minipage} \\
  \hline
  $DMD^{-1}$ & \coly $DMD$     & \colm \begin{minipage}{1.6cm}{$DMD^{-1}M$\\ $= -DM$}\end{minipage} & \coly \begin{minipage}{1.6cm}{\medskip$DMDMD$\\$ = DMD$}\medskip\end{minipage}     & $0$                    \\
  \hline
\end{tabular}
\end{center}

\begin{proof}[Proof of Theorem~\ref{thm:main}]

Notice that in the body of the table only (scalar multiples of) $7$ matrices are used, and the entries are aligned in such a way that two occurrences of the same matrix are either in the same row or in the same column.
At this point we are done proving Theorem~\ref{thm:main}, because the existence of such a pattern gives a simple way to construct a matrix multiplication algorithm as follows.
To multiply matrices $A$ and $B$, represent them in the bases  $\{ D, M, D^{-1}MD, DMD^{-1} \}$ and $\{ D^{-1}, M, D^{-1}MD, DMD^{-1} \}$, respectively:

\begin{equation}\label{eq:AB}
\begin{alignedat}{5}
  X & =  x_1 D      && + x_2 M && + x_3 D^{-1}MD && + x_4 DMD^{-1} \\
  Y & =  y_1 D^{-1} && + y_2 M && + y_3 D^{-1}MD && + y_4 DMD^{-1}
\end{alignedat}
\end{equation}

Note that the $x_i$ are linear forms in the entries of $X$ and the $y_j$ are linear forms in the entries of $Y$.
We expand the product $XY$ and group together summands according to the table:

\begin{equation*}
\begin{array}{r c c c c c}
  XY = & x_1 & \times & y_1 & \times & \colk \id \\
     + & x_2 & \times & (y_1 + y_4) & \times & \colr MD^{-1} \\
     + & x_3 & \times & (y_1 + y_2) & \times & \colg D^{-1}M \\
     + & x_4 & \times & (y_1 + y_3) & \times & \coly DMD     \\
     + & (x_1 - x_4) & \times & y_2 & \times & \colm DM       \\
     + & (x_1 - x_2) & \times & y_3 & \times & \colc MD       \\
     + & (x_1 - x_3) & \times & y_4 & \times & \colb D^{-1}MD^{-1} \\
       &  \uparrow      &       & \uparrow & & \uparrow \\
       &  u_k(X)        &       & v_k(Y)   & & W_k
\end{array}
\end{equation*}
This finishes the proof.
\end{proof}

\begin{remark*}
Taking the trace in \eqref{eq:AB} and using the fact that $M$ and its conjugates are traceless, we see that $\tr(X)=x_1 \tr(D) = -x_1$, and $\tr(Y)=-y_1$.
Thus the first of the 7 summands is $\tr(X)\tr(Y)\id$.
\end{remark*}

\providecommand{\noopsort}[1]{}


\begin{thebibliography}{10}

\bibitem{Alekseyev:96}
Valery~B. Alekseyev.
\newblock Maximal extensions with simple multiplication for the algebra of
  matrices of the second order.
\newblock {\em Discrete Math. Appl.}, 7(1):89--102, 1996.
\newblock \href {http://dx.doi.org/10.1515/dma.1997.7.1.89}
  {\path{doi:10.1515/dma.1997.7.1.89}}.

\bibitem{BILR:17}
Grey Ballard, Christian Ikenmeyer, Joseph~M. Landsberg, and Nick Ryder.
\newblock The geometry of rank decompositions of matrix multiplication {II}: $3
  \times 3$ matrices.
\newblock Preprint arXiv:1801.00843, arXiv, 2018.
\newblock \href {http://arxiv.org/abs/1801.00843} {\path{arXiv:1801.00843}}.

\bibitem{Blaeser:13}
Markus Bl{\"a}ser.
\newblock Fast matrix multiplication.
\newblock {\em Theory Comput. Grad. Surv.}, 5:1--60, 2013.
\newblock \href {http://dx.doi.org/10.4086/toc.gs.2013.005}
  {\path{doi:10.4086/toc.gs.2013.005}}.

\bibitem{Brent:70}
Richard~P. Brent.
\newblock Algorithms for matrix multiplication.
\newblock Tech. Report STAN-CS-70-157, Stanford University, Department of
  Computer Science, 1970.
\newblock \href {http://dx.doi.org/10.21236/ad0705509}
  {\path{doi:10.21236/ad0705509}}.

\bibitem{Brockett-Dobkin:73}
Roger~W. Brockett and David Dobkin.
\newblock On the optimal evaluation of a set of bilinear forms.
\newblock In {\em Proc. 5th ACM STOC}, pages 88--95, 1973.
\newblock \href {http://dx.doi.org/10.1145/800125.804039}
  {\path{doi:10.1145/800125.804039}}.

\bibitem{Buechi-Clausen:85}
Werner B{\"u}chi and Michael Clausen.
\newblock On a class of primary algebras of minimal rank.
\newblock {\em Linear Algebra Appl.}, 69:249--268, 1985.
\newblock \href {http://dx.doi.org/10.1016/0024-3795(85)90080-1}
  {\path{doi:10.1016/0024-3795(85)90080-1}}.

\bibitem{BCS:97}
Peter B{\"u}rgisser, Michael Clausen, and M.~Amin Shokrollahi.
\newblock {\em Algebraic Complexity Theory}, volume 315 of {\em {Grundlehren}
  der mathematischen {Wissenschaften}}.
\newblock Springer, Berlin, 1997.
\newblock \href {http://dx.doi.org/10.1007/978-3-662-03338-8}
  {\path{doi:10.1007/978-3-662-03338-8}}.

\bibitem{Burichenko:14}
Vladimir~P. Burichenko.
\newblock On symmetries of the {Strassen} algorithm.
\newblock Preprint arXiv:1408.6273, arXiv, 2014.
\newblock \href {http://arxiv.org/abs/1408.6273} {\path{arXiv:1408.6273}}.

\bibitem{Chatelin:86}
Philippe Chatelin.
\newblock On transformations of algorithms to multiply $2 \times 2$ matrices.
\newblock {\em Inf. Process. Lett.}, 22(1):1--5, 1986.
\newblock \href {http://dx.doi.org/10.1016/0020-0190(86)90033-5}
  {\path{doi:10.1016/0020-0190(86)90033-5}}.

\bibitem{CILO:17}
Luca Chiantini, Christian Ikenmeyer, Joseph~M. Landsberg, and Giorgio
  Ottaviani.
\newblock The geometry of rank decompositions of matrix multiplication {I}: $2
  \times 2$ matrices.
\newblock {\em Exp. Math.}, Advance online publication, 2017.
\newblock \href {http://dx.doi.org/10.1080/10586458.2017.1403981}
  {\path{doi:10.1080/10586458.2017.1403981}}.

\bibitem{Clausen:88}
Michael Clausen.
\newblock {\em {Beitr{\"a}ge} zum {Entwurf} schneller
  {Spektraltransformationen}}.
\newblock Habilitationsschrift, Universit{\"a}t Karlsruhe, 1988.

\bibitem{Dumas-Pan:16}
Jean-Guillaume Dumas and Victor~Y. Pan.
\newblock Fast matrix multiplication and symbolic computation.
\newblock Preprint arXiv:1612.05766, arXiv, 2016.
\newblock \href {http://arxiv.org/abs/1612.05766} {\path{arXiv:1612.05766}}.

\bibitem{Fiduccia:72}
Charles~M. Fiduccia.
\newblock On obtaining upper bounds on the complexity of matrix multiplication.
\newblock In {\em Complexity of Computer Computations}, pages 31--40, 1972.
\newblock \href {http://dx.doi.org/10.1007/978-1-4684-2001-2_4}
  {\path{doi:10.1007/978-1-4684-2001-2_4}}.

\bibitem{Gastinel:71}
No{\"e}l Gastinel.
\newblock Sur le calcul des produits de matrices.
\newblock {\em Numer. Math.}, 17(3):222--229, 1971.
\newblock \href {http://dx.doi.org/10.1007/BF01436378}
  {\path{doi:10.1007/BF01436378}}.

\bibitem{Gates-Kreinovich:01}
Ann~Q. Gates and Vladik Kreinovich.
\newblock Strassen's algorithm made (somewhat) more natural: A pedagogical
  remark.
\newblock {\em Bull. EATCS}, 73:142--145, 2001.
\newblock URL: \url{https://digitalcommons.utep.edu/cs_techrep/502/}.

\bibitem{Grochow-Moore:16}
Joshua~A. Grochow and Cristopher Moore.
\newblock Matrix multiplication algorithms from group orbits.
\newblock Preprint arXiv:1612.01527, arXiv, 2016.
\newblock \href {http://arxiv.org/abs/1612.01527} {\path{arXiv:1612.01527}}.

\bibitem{Grochow-Moore:17}
Joshua~A. Grochow and Cristopher Moore.
\newblock Designing {Strassen's} algorithm.
\newblock Preprint arXiv:1708.09398, arXiv, 2017.
\newblock \href {http://arxiv.org/abs/1708.09398} {\path{arXiv:1708.09398}}.

\bibitem{deGroote:78b}
Hans~F. {\noopsort{Groote}}{de Groote}.
\newblock On varieties of optimal algorithms for the computation of bilinear
  mappings {II}: Optimal algorithms for $2 \times 2$-matrix multiplication.
\newblock {\em Theor. Comp. Sci.}, 7(2):127--184, 1978.
\newblock \href {http://dx.doi.org/10.1016/0304-3975(78)90045-2}
  {\path{doi:10.1016/0304-3975(78)90045-2}}.

\bibitem{HRMvdG:17}
Jianyu Huang, Leslie Rice, Devin~A. Matthews, and Robert~A. van~de Geijn.
\newblock Generating families of practical fast matrix multiplication
  algorithms.
\newblock In {\em Proc. IPDPS 2017}, pages 656--667, 2017.
\newblock \href {http://dx.doi.org/10.1109/IPDPS.2017.56}
  {\path{doi:10.1109/IPDPS.2017.56}}.

\bibitem{Lafon:75}
Jean-Claude Lafon.
\newblock Optimum computation of $p$ bilinear forms.
\newblock {\em Linear Algebra Appl.}, 10(3):225--240, 1975.
\newblock \href {http://dx.doi.org/10.1016/0024-3795(75)90071-3}
  {\path{doi:10.1016/0024-3795(75)90071-3}}.

\bibitem{LeGall:14}
Fran\c{c}ois Le~Gall.
\newblock Powers of tensors and fast matrix multiplication.
\newblock In {\em Proc. ISSAC 2014}, pages 296--303, 2014.
\newblock \href {http://dx.doi.org/10.1145/2608628.2608664}
  {\path{doi:10.1145/2608628.2608664}}.

\bibitem{Makarov:75}
Oleg~M. Makarov.
\newblock The connection between two multiplication algorithms.
\newblock {\em USSR Comput. Math. Math. Phys.}, 15(1):218--223, 1975.
\newblock \href {http://dx.doi.org/10.1016/0041-5553(75)90149-4}
  {\path{doi:10.1016/0041-5553(75)90149-4}}.

\bibitem{Pan:72}
Victor~Y. Pan.
\newblock О схемах вычисления произведений
  матриц и обратной матрицы [{On} algorithms for matrix
  multiplication and inversion].
\newblock {\em Усп. мат. наук}, 27(5(167)):249--250, 1972.
\newblock Translation available in~\cite{Pan:14}.
\newblock URL: \url{http://mi.mathnet.ru/umn5125}.

\bibitem{Pan:14}
Victor~Y. Pan.
\newblock Better late than never: Filling a void in the history of fast matrix
  multiplication and tensor decompositions.
\newblock Preprint arXiv:1411.1972, arXiv, 2014.
\newblock \href {http://arxiv.org/abs/1411.1972} {\path{arXiv:1411.1972}}.

\bibitem{Pan:17}
Victor~Y. Pan.
\newblock Fast matrix multiplication and its algebraic neighbourhood.
\newblock {\em Sb. Math.}, 208(11):1661--1704, 2017.
\newblock \href {http://dx.doi.org/10.1070/SM8833} {\path{doi:10.1070/SM8833}}.

\bibitem{Paterson:74}
Mike Paterson.
\newblock Complexity of product and closure algorithms for matrices.
\newblock In {\em Proc. ICM 1974}, volume~2, pages 483--489, 1974.
\newblock URL:
  \url{https://www.mathunion.org/fileadmin/ICM/Proceedings/ICM1974.2/ICM1974.2.ocr.pdf#page=491}
  [cited 2018-02-03].

\bibitem{Paterson:09}
Mike Paterson.
\newblock Strassen symmetries.
\newblock Presentation at Leslie Valiant's 60th birthday celebration,
  30.05.2009, Bethesda, Maryland, USA, 2009.
\newblock URL: \url{https://www.cis.upenn.edu/~mkearns/valiant/paterson.ppt}
  [cited 2018-02-03].

\bibitem{Strassen:69}
Volker Strassen.
\newblock Gaussian elimination is not optimal.
\newblock {\em Numer. Math.}, 13(4):354--356, 1969.
\newblock \href {http://dx.doi.org/10.1007/BF02165411}
  {\path{doi:10.1007/BF02165411}}.

\bibitem{Yuval:78}
Gideon Yuval.
\newblock A simple proof of {Strassen}'s result.
\newblock {\em Inf. Process. Lett.}, 7(6):285--286, 1978.
\newblock \href {http://dx.doi.org/10.1016/0020-0190(78)90018-2}
  {\path{doi:10.1016/0020-0190(78)90018-2}}.

\end{thebibliography}
\end{document}